\documentclass[reqno]{amsart}
\newtheorem{thm}{Theorem}[section] 
 \newtheorem{lem}[thm]{Lemma} 
\newtheorem{prop}[thm]{Proposition} \theoremstyle{definition} 
\newtheorem{defn}[thm]{Definition} \theoremstyle{remark} 
\newtheorem{rem}[thm]{Remark} 
 
\numberwithin{equation}{section} 
\usepackage{amsfonts}
 \usepackage{amsmath}
 \usepackage{amssymb}
 \begin{document}

\title[Singularities of the scattering kernel] {Singularities of the scattering kernel related to trapping rays}\author{Vesselin Petkov}\address{Institut de Math\'ematiques de Bordeaux, 351 Cours de la Lib\'eration, 33405 Talence, France}\email{petkov@math.u-bordeaux1.fr}
\author{Luchezar Stoyanov}\address{School of Mathematics and Statistics, University of Western Australia, 35 Stirling Hwy, Crawley 6009, Western Australia}\email{stoyanov@maths.uwa.edu.au}
\date{}

\dedicatory{Dedicated to Ferruccio Colombini on the occasion of his 60th birtday}
\newcommand{\h}{{\mathcal H}}
\newcommand{\R}{{\mathbb R}}
\newcommand{\N}{{\mathbb N}}
\newcommand{\C}{{\mathbb C}}
\newcommand{\F}{{\mathcal F}}
\newcommand{\Oo}{{\mathcal O}}
\newcommand{\K}{{\mathcal K}}
\newcommand{\D}{{\mathcal D}}
\newcommand{\G}{{\mathcal G}}
\newcommand{\Hh}{{\mathcal H}}
\newcommand{\Z}{{\mathbb Z}}
\newcommand{\Q}{{\mathbb Q}}
\newcommand{\U}{{\mathcal U}}
\newcommand{\A}{{\mathbb A}}
\newcommand{\Ss}{{\mathbb S}}
\renewcommand{\Re}{\mathop{\rm Re}\nolimits}
\renewcommand{\Im}{\mathop{\rm Im}\nolimits}
\begin{abstract}

An obstacle $K \subset \R^n,\: n \geq 3,$ $n$ odd, is called trapping if there exists at least one generalized bicharacteristic $\gamma(t)$ of the wave equation staying in a neighborhood of $K$ for all $t \geq 0.$
We examine the singularities of the scattering kernel $s(t, \theta, \omega)$ defined as the Fourier transform of the scattering amplitude $a(\lambda, \theta, \omega)$ related to the Dirichlet problem for the wave equation in $\Omega = \R^n \setminus K.$ We prove that if $K$ is trapping and $\gamma(t)$ is non-degenerate, then there exist reflecting $(\omega_m, \theta_m)$-rays $\delta_m,\: m \in \N,$ with sojourn times $T_m \to +\infty$ as $m \to \infty$, so that $-T_m \in {\rm sing}\:{\rm supp}\: s(t, \theta_m, \omega_m),\: \forall m \in \N$. We apply this property to study the behavior  of the scattering amplitude in $\C$.\\
\end{abstract}

\maketitle

\numberwithin{equation}{section}

\def\CC{{\mathcal C}}
\def\lap{\bigtriangleup}
\def\ra{\rangle}
\def\la{\langle}
\def\TO{\dot{T}^*(\Omega)}
\def\ff{{\mathcal F}}
\def\rr{{\mathcal R}}
\def\ot{(\omega, \theta)}
\def\e{\varepsilon}
\def\phi {\varphi}
\def \la {{\lambda}}
\def \a {{\alpha}}
\def\sn{{\mathbb S}^{n-1}}
\def\ssn{{\sn}\times {\sn}}
\def\pp{P_{+}}
\def\ppm{P_{-}}
\def\ppr{P_{+}^{\rho}}
\def\ppmr{P_{-}^{\rho}}
\def\hl{{\mathcal H}_{{\rm loc}}}
\def\rp{{\rm Res}\:P(t)}
\def\ggi{{\mathcal G}^{(i)}}
\def\td{\tilde{d}}
\def\cS{\check{S}}
\def\ii{{\mathcal I}}
\def\mm{{\mathcal M}}
\def\cp{\check{p}}
\def\pr{{\rm pr}}
\def\uu{{\mathcal U}}
\def\qq{{\mathcal Q}}
\def\ovo{\overset{\circ}\Omega}
\def\ii{{\bf i}}
\def\la{\langle}
\def\ra {\rangle}

\section{Introduction}

Let $K \subset \{x \in \R^n,\:|x| \leq \rho\},\: n \geq 3,\:n$ odd, be a bounded domain with $C^{\infty}$ boundary $\partial K$ 
and connected complement $\Omega = \overline{\R^n \setminus K}.$ Such $K$ is called an {\it obstacle}
in $\R^n$. In this paper we consider the Dirichlet problem for the wave equation however in a similar way one can deal with other boundar value problems. Given two directions $(\theta, \omega) \in \ssn$, consider the {\em outgoing solution} $ v_s(x, \lambda)$ of the problem
$$
\begin{cases} (\Delta + \lambda^2) v_s = 0 \:\: {\rm in} \:\: \ovo,\\
v_s + e^{- \ii\lambda \langle x,  \omega \rangle} = 0\:\:{\rm on} \: \partial K, \end{cases}
$$
satisfying the so called $(\ii \lambda)$ - outgoing Sommerfeld radiation condition:
\[ v_s(r\theta, \lambda) = \frac{e^{-\ii \lambda r}}{r^{(n-1)/2}} \Bigl(a(\lambda, \theta, \omega) + 
{\mathcal O} \Bigl(\frac{1}{r} \Bigr)\Bigr), \:\: x = r\theta,\: {\rm as}\: |x| = r  \longrightarrow \infty \,.\] 
 The leading term $a(\lambda, \theta, \omega)$ is called {\it scattering amplitude} and we have the following representation
\begin{equation}
 a(\lambda, \theta, \omega) =  \frac{(\ii\lambda)^{(n-3)/2}}{2(2\pi)^{(n-1)/2}} \int_{\partial K} 
\Bigl( \ii \lambda \langle \nu(x), \theta \rangle e^{\ii\lambda \langle x , \theta - \omega \rangle } - 
e^{\ii\lambda \langle x, \theta \rangle } \frac{\partial v_s}{\partial \nu} (x, \lambda) \Bigr) dS_x \,, 
\end{equation}
where  $\langle \bullet , \bullet \rangle $ denotes the inner product in $\R^n$ and $\nu(x)$ is the unit normal to $x \in \partial K$ pointing into $\Omega$ (see \cite{Ma2}, \cite{P}).\\

Throughout this note we assume that $\theta \neq \omega.$ The {\it scattering kernel} $s(t, \theta, \omega)$ is defined as the Fourier transform of the scattering amplitude
\[ s(t, \theta, \omega) = {\mathcal F}_{\lambda \to t} \Bigl(\Bigl(\frac{\ii\lambda}{2 \pi}\Bigr)^{(n-1)/2}  \overline{ a(\lambda, \theta, \omega)}\Bigr) \,, \]
where $\Bigl({\mathcal F}_{\lambda \to t} \varphi\Bigr)(t) = (2\pi)^{-1}\int e^{\ii t\lambda} \varphi(\lambda) d\lambda$ 
for functions $\varphi \in {\mathcal S}(\R).$
Let $V(t, x;\omega)$ be the solution of the problem
$$ \begin{cases}(\partial^2_t - \Delta) V =  0\:\: {\rm in} \:\: \R \times \ovo,\\
V = 0\:\:{\rm on} \: \R \times \partial K,\\
V\vert_{t < - \rho} = \delta(t - \la x, \omega \ra). \end{cases}$$
Then we have
$$ s(\sigma, \theta, \omega) = (-1)^{(n+1)/2} 2^{-n}\pi^{1-n} \int_{\partial K}  \partial^{n-2}_t \partial_{\nu} 
V (\la x, \theta \ra -\sigma, x; \omega) dS_x\; ,$$
where the integral is interpreted in the sense of distributions.\\

 The singularities of $s(t, \theta, \omega)$ with respect to $t$ can be observed since at these times we have some non negligible picks of the scattering amplitude. For example, if $K$ is strictly convex, for fixed $\theta \neq \omega$ we have only one singularity at $t = -T_{\gamma}$ related to the sojourn time of the unique $(\omega, \theta)$-reflecting ray $\gamma$ (see \cite{Ma1}). For general non-convex obstacles the geometric situation is much more complicated since we have different type of rays incoming with direction $\omega$ and outgoing in direction $\theta$ for which an asymptotic solution related to the rays is impossible to construct. 
In many problems, such as those concerning local decay of energy, behavior of the cut-off resolvent of the Laplacian, the existence of resonances etc. the difference between non-trapping and trapping obstacles is quite significant. In recent years many authors studied mainly trapping obstacles with some very special geometry and the case of several strictly convex disjoint obstacles has been investigated both from mathematical and numerical analysis point of view.\\

In this work our purpose is the study the obstacles having at least one $(\omega, \theta)$-trapping ray $\gamma$ which in general could be non-reflecting (see Section 2 for the definition of an $(\omega, \theta)$-ray). 
No assumptions are made on the geometry of the obstacle outside some small neighborhood of $\gamma$ and no information is required about other possible $(\omega, \theta)$-rays. Our aim is to examine if the existence of $\gamma$ may create an {\bf infinite number} of delta type singularities $T_m \to \infty$ of $s(-t, \theta_m, \omega_m)$, in contrast to the non-trapping case where $s(t, \theta, \omega)$ is $C^{\infty}$ smooth for $|t| \geq T_0 > 0$ and all $(\theta, \omega) \in \ssn.$ On the other hand, it is important to stress that the scattering amplitude and the scattering kernel are global objects and their behavior depends on all $(\omega, \theta)$-rays so any type of cancellation of singularities may occur. The existence of a trapping ray influences the singularities of $s(t, \theta, \omega)$ if we assume that $\gamma$ is non-degenerate which is a local condition (see Section 3). Thus our result says that from the scattering data related to the singularities of $s(t, \theta, \omega)$ we can ``hear'' whether $K$ is trapping or not.

The proof of our main result is based on several previous works \cite{P}, \cite{PS1}, \cite{PS2}, \cite{PS3}, \cite{St1}, and our purpose here is to show how the results of these works imply the existence of an infinite number of singularities. The reader may consult \cite{PS4} for a survey on the results mentioned above.

\section{Scattering kernel}

We start with the definition of the so called reflecting $\ot$-rays.
Given two directions $\ot \in \ssn$,  consider a curve $\gamma \in \Omega$ having the form
$$ \gamma = \cup_{i = 0}^{m} l_i, \:\: m \geq 1,$$
where $l_i = [x_i, x_{i+1}]$ are finite segments for $i = 1,...,m-1,\: x_i \in \partial K$, and $l_0$  (resp. $l_m$) 
is the infinite segment starting at $x_1$ (resp. at $x_m$) and having direction $-\omega$ (resp. $\theta$). 
The curve $\gamma$ is called a {\it reflecting} $\ot$-ray in $\Omega$ if for $i = 0,1,...,m-1$ the 
segments
$l_i$ and $l_{i+1}$ satisfy the law of reflection at $x_{i+1}$ with respect to $\partial K$. The points 
$x_1,...,x_m$ are called {\it reflection points} of $\gamma$ and this ray is called {\it ordinary reflecting} 
 if $\gamma$ has no segments tangent to $\partial K.$\\

Next, we define two notions related to $\ot$-rays.  Fix an arbitrary open ball $U_0$ with radius 
$a > 0$ containing $K$ and  for $\xi \in \sn$ introduce the hyperplane $Z_{\xi}$ orthogonal to $\xi$, tangent to $U_0$ 
 and such that $\xi$ is pointing into the interior of the open half space $H_{\xi}$ with boundary $Z_{\xi}$ containing $U_0$. 
Let
$\pi_{\xi} : \R^n \longrightarrow Z_{\xi}$ be the orthogonal projection. For a reflecting $(\omega, \theta)$-ray
$\gamma$ in $\Omega$ with successive reflecting points $x_1,..., x_m$ the {\it sojourn time} $T_{\gamma}$ of $\gamma$
is defined by
\[ T_{\gamma} = \|\pi_{\omega}(x_1) - x_1 \| + \sum_{i=1}^{m-1} \|x_i - x_{i+1}\| + \|x_m - \pi_{-\theta}(x_m) \|
 - 2a \,. \]
Obviously, $T_{\gamma} + 2a$ coincides with the length of the part of $\gamma$ that lies in $H_{\omega} \cap  H_{-\theta}$.
The sojourn time $T_{\gamma}$ does not depend on the choice of the ball $U_0$ and
\[ T_{\gamma} = \la x_1, \omega \ra + \sum_{i=1}^{m-1} \|x_i - x_{i+1}\| - \la x_m, \theta \ra \,. \]

\begin{figure}
\begin{center}
\ifx\cs\tempdima \newdimen \tempdima\fi
\ifx\cs\tempdimb \newdimen \tempdimb\fi
\ifx\cs\jpdrawx \newcount \jpdrawx\fi
\ifx\cs\jpspace \def\jpspace{ } \fi
\jpdrawx=100 \tempdima=   1.00cm \tempdimb= 1.0cm
\ifdim \unitlength = 1pt
\unitlength=  1.00cm
\fi
\ifdim \unitlength < 0cm
\multiply \unitlength by -10 \divide \unitlength by 125 \fi
\ifdim \unitlength > \tempdima \tempdima=\unitlength \fi
\ifdim \unitlength < \tempdima \tempdima=\unitlength \fi
\multiply \jpdrawx by \tempdima \divide \jpdrawx by \tempdimb
\setlength{\unitlength}{\tempdima}
\begin{picture}(  12.5,  10.0)( 0.986, 2.958)
\thinlines
\large
\put( 0.986, 2.958){
}
\scriptsize
\put(5.13,5.81){$\partial K$}
\put(6.37,3.38){$U_0$}
\put(11.69,7.79){$H_{-\theta}$}
\put(11.74,9.92){$Z_{-\theta}$}
\put(7.31,10.64){$\theta(u)$}
\put(8.90,9.90){$\theta$}
\put(5.44,11.79){$u_{\gamma}$}
\put(4.83,11.70){$u$}
\put(2.58,10.44){$\omega$}
\put(1.32,11.12){$Z_{\omega}$}
\put(5.07,7.63){$x_1$}
\put(7.82,7.34){$x_m$}
\put(8.09,7.93){$x_m(u)$}
\put(4.41,8.18){$x_1(u)$}
\end{picture}
\setlength{\unitlength}{1cm}

\end{center}
\caption{}
\end{figure}

\medskip

Given an ordinary reflecting $(\omega, \theta)$-ray $\gamma$ set $u_{\gamma} = \pi_{\omega}(x_1).$ Then there exists a 
small neighborhood $W_{\gamma}$ of $u_{\gamma}$ in $Z_{\omega}$ such that for every $u \in W_{\gamma}$ there is an
unique direction $\theta(u) \in \sn$ and points $x_1(u),...,x_m(u)$ which are the successive reflection points of
a reflecting $(u, \theta(u))$-ray in $\Omega$ with $\pi_{\omega}(x_1(u)) = u$ (see Figure 1). We obtain a smooth map
\[ J_{\gamma} : W_{\gamma} \ni u \longrightarrow \theta(u) \in \sn \, \]
and $dJ_{\gamma}(u_{\gamma})$ is called a {\it differential cross section} related to $\gamma$. We say that 
$\gamma$ is  {\it non-degenerate} if
\[ \det d J_{\gamma} (u_{\gamma}) \neq 0 \,. \]
The notion of sojourn time as well as that of differential cross section are well known in the  physical literature and the definitions given above are due to Guillemin \cite{G}.\\

For non-convex obstacles there exist $\ot$-rays with some tangent and/or  gliding segments. 
To give a precise definition one has to involve the  generalized bicharacteristics of the 
operator $\square = \partial_t^2 - \Delta_x$ defined as the trajectories of the generalized Hamilton flow ${\mathcal F}_t$
in $\Omega$ generated by the symbol $\sum_{i=1}^n \xi_i^2 - \tau^2$ of $\square$ (see \cite{MS} for a precise
definition). In general,  ${\mathcal F}_t$ is not smooth and in some cases there may exist two different integral
curves issued from the same point in the phase space (see \cite{T} for an example). To avoid this situation in the following 
we assume that the following generic condition is satisfied.\\

$({\mathcal G})$ \:\:\:\:\:\:\:\:\: If for $(x, \xi) \in T^*(\partial K)$ the normal curvature of $\partial K$ 
vanishes of infinite order in direction $\xi$, then $\partial K$ is convex at $x$ in direction $\xi.$\\

Given $\sigma = (x,\xi) \in  T^*(\Omega) \setminus \{0\} = \TO$, there exists a unique generalized bicharacteristic 
$(x(t),\xi(t)) \in  \TO$ such that $x(0) = x,\:\: \xi(0) = \xi$ and we define $\ff_t(x,\xi) = (x(t), \xi(t))$ for all $t\in \R$(see \cite{MS}). We obtain a flow $\ff_t : \TO \longrightarrow \TO$  which is called
the {\it generalized geodesic flow} on $\TO$. It is clear, that this flow leaves the {\it cosphere bundle} $S^*(\Omega)$ invariant. The flow $\ff_t$ is  discontinuous at points of transversal reflection at $\dot{T}^*_{\partial K}(\Omega)$ and to make it continuous, consider the {\it quotient space} $ \TO/\sim$ of $\TO$ with 
respect to the following equivalence relation: $\rho\sim \sigma$  if and only if 
$\rho = \sigma$ or $\rho, \sigma \in T^*_{\partial K}(\Omega)$ and either $\lim_{t \nearrow 0} \ff_t(\rho) = \sigma$
or $\lim_{t \searrow 0} \ff_t(\rho) = \sigma$.   Let $\Sigma_b$ be the image of $S^*(\Omega)$
in $\TO/\sim$. The set $\Sigma_b$ is called the {\it compressed characteristic set}.  Melrose and Sj\"{o}strand (\cite{MS})
proved that the natural projection of $\ff_t$ on $\TO/\sim$ is continuous.

\medskip

Now a curve $\gamma = \{ x(t) \in \Omega: \: t \in \R \}$ is called an {\it $\ot$-ray} if there 
exist real numbers $t_1 < t_2$ so that
\[ \hat{\gamma}(t) = (x(t), \xi(t)) \in S^*(\Omega) \, \]
is a {\it generalized bicharacteristic} of $\square$ and
$$\xi(t) = \omega \:\:{\rm for} \:\: t \leq t_1, \:\:\xi(t) = \theta \:\:{\rm for} \:\: t \geq t_2 ,$$
provided that the time $t$ increases when we move along $\hat{\gamma}.$ Denote by 
${\mathcal L}_{\ot} (\Omega)$  the {\it set of all $\ot$-rays}
in $\Omega.$ The {\it sojourn time} $T_{\delta}$ of 
$\delta \in {\mathcal L}_{\ot} (\Omega)$ is defined as the length of the part of $\delta$ lying in
$H_{\omega} \cap H_{-\theta}.$\\

It was proved in \cite{P1}, \cite{CPS} (cf. also  Chapter 8 in \cite{PS1} and \cite{M2}) that
for $\omega \neq \theta$ 
we have
\begin{equation} \label{eq:2.2}
\mbox{sing supp}_t\:\: s(t,\theta,\omega) \subset \{ - T_{\gamma} :
\gamma\in {\mathcal L}_{\ot}(\Omega) \}.
\end{equation}
This relation was established for convex obstacles by Majda \cite{Ma2} and for some Riemann surfaces by Guillemin \cite{G}. The proof in \cite{P1}, \cite{CPS} deals with general obstacles and is based on the results in \cite{MS} concerning 
propagation of singularities.

In analogy with the well-known Poisson relation for the Laplacian on Riemannian manifolds, ($\ref{eq:2.2}$) is called {\it the Poisson relation for the scattering kernel}, while the set of all $T_{\gamma}$, where $\gamma \in {\mathcal L}_{\ot}(\Omega)$, 
$\ot \in \ssn$, is called the {\it scattering length spectrum} of $K$.\\

To examine the behavior of $s(t, \theta, \omega)$ near singularities, assume that $\gamma $ is a fixed
{\it non-degenerate ordinary reflecting} $\ot$-ray such that
\begin{equation} \label{eq:2.3}
 T_{\gamma} \neq T_{\delta} \:\:{\rm for}\:\:{\rm every}\:\: \delta \in {\mathcal L}_{\ot} 
(\Omega) \setminus \{\gamma \}.
\end{equation}
By using the continuity of the generalized Hamiltonian flow, it is easy to show that
\begin{equation} \label{eq:2.4}
(-T_{\gamma} - \epsilon, -T_{\gamma} + \epsilon) \cap {\rm sing}\:\:{\rm supp}_t\:\: s(t,\theta, \omega) = 
\{-T_{\gamma}\}
\end{equation}
for $\epsilon > 0$ sufficiently small. For strictly convex obstacles and $\omega \neq \theta$ every $\ot$-ray is non-degenerate and $(\ref{eq:2.4})$ is obviously satisfied. For general non-convex obstacles one needs to establish some global properties of $\ot$-rays and choose $\ot$ so that $(\ref{eq:2.4})$ holds. The singularity of $s(t,\theta, \omega)$ at $t = -T_{\gamma}$ can be investigated by using a global  construction of an asymptotic solution as a Fourier integral operator (see \cite{GM}, \cite{P1} and Chapter 9 in \cite{PS1}), and we have the following

\medskip

\begin{thm}
{\rm (\cite{P1})} Let $\gamma$ be a non-degenerate ordinary reflecting $\ot$-ray and let $\omega \neq \theta$. Then under the assumption {\rm (\ref{eq:2.4})} we have
\begin{equation}
-T_{\gamma} \in {\rm sing}\:\:{\rm supp}_t\:\: s(t, \theta, \omega)
\end{equation} 
and for $t$ close to $-T_{\gamma}$ the scattering kernel has the form 
\begin{equation}
s(t, \theta, \omega) = \Bigl(\frac{1}{2 \pi \ii}\Bigr)^{(n-1)/2}(-1)^{m_{\gamma} - 1} \exp\Bigl( \ii\frac{\pi}{2} 
\beta_{\gamma}\Bigr) 
\end{equation} 
$$ \times \Bigl | \frac{\det d J_{\gamma}(u_{\gamma}) \langle \nu(q_1), \omega \rangle}{\langle\nu(q_m), \theta \rangle}\Bigr |^{-1/2} 
\delta^{(n-1)/2}(t + T_{\gamma}) +  {\rm lower} \:\: {\rm order}\:\:{\rm singularities}.$$ 
Here $m_{\gamma}$ is the number of reflections of $\gamma$, $q_1$ $($resp. $q_m)$ is the first $($resp. the last$)$ 
reflection point of $\gamma$ and $\beta_\gamma \in \Z$. 
\end{thm}

\medskip

 For strictly convex obstacles we have 
$$m_{\gamma_{+}} = 1,\:\beta_{\gamma_{+}} = -\frac{n-1}{2}, \:\:q_1 = q_m,$$
$ \theta - \omega$ is parallel to $\nu(q_1)$ and 
$$|\det dJ _{\gamma_{+}}(u_{\gamma_{+}})| = 4 |\theta - \omega|^{n-3} {\mathcal K}(x_{+}),$$
where $\gamma_{+}$ is the unique $\ot$-reflecting ray at $x_{+}$, $u_{+}$ is the corresponding point on $Z_{\omega}$ and ${\mathcal K}(x_{+})$ is the Gauss curvature at $x_{+}.$ Thus we obtain the result of Majda \cite{Ma1} (see also \cite{Ma2}) describing the leading singularity at $-T_{\gamma_{+}}.$\\

To obtain an equality in the Poisson relation $(\ref{eq:2.2})$, one needs to know that every $\ot$-ray produces a singularity. To
achieve this, a natural way to proceed would be  to ensure that the properties $(\ref{eq:2.3}),\:( \ref{eq:2.4})$ hold. It is clear, that these properties depend on the global behavior of the $(\omega, \theta)$-rays in the exterior of the obstacle, and in this regard the existence 
of $\ot$-rays with tangent or gliding segments leads to considerable difficulties. Moreover, different ordinary reflecting rays could produce singularities which mutually cancel. By using the properties of $\ot$-rays established in  \cite{PS2}, \cite{PS3}, as well as the fact that for almost all directions $(\omega, \theta)$, the $\ot$-rays are ordinary reflecting (see \cite{St1}), the following was derived in \cite{St1}:

\medskip

\begin{thm} \mbox{\rm (\cite{St1})} There exists a subset $\rr$ of full Lebesgue measure in $\ssn$ such that
for each $(\omega, \theta) \in \rr$ the only $(\omega, \theta)$-rays in $\Omega$ are ordinary reflecting $(\omega, \theta)$-rays and
$${\rm sing} \:\: {\rm supp}_t\: s (t,\theta,\omega) =
\{ -T_{\gamma} : \gamma \in {\mathcal L}_{\omega,\theta} (\Omega)\}\;.$$
\end{thm}

This result is the basis for  several interesting inverse scattering results (see \cite{St2}, \cite{St3}).

\section{Trapping obstacles}

\def\sb{\Sigma_b}
\def\si{\Sigma_{\infty}}

Given a generalized bicharacteristic $\gamma$ in $S^*(\Omega)$, its projection $\tilde{\gamma} = \sim (\gamma)$ in
$\Sigma_b$ is called a {\it compressed generalized bicharacteristic}. Let $U_0$ be an open ball containing $K$ and let $C$ be its boundary sphere. For an arbitrary point $z = (x ,\xi) \in \sb$, consider the compressed  generalized bicharacteristic
$$\gamma_z(t) = (x(t),\xi(t)) \in \sb$$
parametrized by the time $t$ and passing through $z$ for $t = 0.$ Denote by $T(z) \in \R^+ \cup \infty$ the maximal $T > 0$ such that $x(t) \in \uu_0$ for $0 \leq t \leq T(z).$ The so called {\it trapping set} is defined by
$$\si = \{(x,\xi) \in \sb:\: x \in C,\: T(z) = \infty \}\;.$$
It follows from the continuity of the compressed generalized Hamiltonian flow that  the trapping set $\si$ is closed in $\Sigma_b$. For simplicity, in the following the compressed generalized bicharacteristics will be called simply generalized ones. The
obstacle  $K$ is called {\it trapping} if $\si \neq \emptyset$, i.e. when 
there exists at least one point $(\hat{x}, \hat{\xi}) \in C \times \sn$ such that the (generalized) trajectory issued from $(\hat{x}, \hat{\xi})$ stays in $U_0$ for all $t \geq 0$.  This provides some information about the behavior of the rays issued from the points $(y, \eta)$ sufficiently close to $(\hat{x}, \hat{\xi})$, however in general it does not yield any information about the geometry of $\ot$-rays. 

Now for every trapping obstacle we have the following

\begin{thm} \mbox{\rm (\cite{PS2}, \cite{PS4})} \label{3.1} 
Let the obstacle $K$ be trapping and satisfy the condition $({\mathcal G})$. Then there exists a sequence of ordinary reflecting $(\omega_m, \theta_m)$-rays  $\gamma_m$ with sojourn times $T_{\gamma_m} \longrightarrow \infty.$
\end{thm}

To prove this we use the following 

\begin{prop} \mbox{\rm (\cite{LP}, \cite{St1})} \label{3.2}
The set of points $(x, \xi) \in S^*_C(\Omega)= \{(x, \xi) \in T^*(\Omega):\: x \in C,\: |\xi| = 1\}$ such that the trajectory $\{\ff_t(x, \xi): \: t \geq 0\}$ issued from $(x, \xi)$ is bounded has Lebesgue measure zero in  $S^*_C(\Omega).$
\end{prop}

\begin{proof} Assume $K$ is trapping and satisfies the condition $({\mathcal G})$.
We will establish the existence of  $(\omega, \theta_m)$-rays with sojourn times $T_m \to \infty$ 
for some $\omega \in \sn$ suitably fixed. It is easy to see that $\sb \setminus \si \neq \emptyset$. 
Since $K$ is trapping, we have $\si \neq \emptyset$, so the boundary $\partial \si$ of $\si$ in $\sb$ is not empty. 
Fix an arbitrary $\hat{z} \in \partial \si$ and take an arbitrary sequence $z_m = (0, x_m, 1, \xi_m) \in \sb$, so that 
$z_m \notin \si$ for every $m \in \N$ and $z_m \longrightarrow \hat{z}.$ Consider the compressed generalized 
bicharacteristics $\delta_m = (t, x_m(t), 1, \xi_m(t))$ passing through $z_m$ for $t = 0$ with sojourn times $T_{z_m} < \infty.$ If the sequence $\{T_{z_m}\}$ is bounded, one gets a contradiction with the fact that $\hat{z} \in \si.$ 
Thus, $\{T_{z_m}\}$ is unbounded, and replacing the sequence $\{ z_m\}$ by an appropriate subsequence of its, we may 
assume that  $T_{z_m} \longrightarrow + \infty$. Setting 
$$y_m = x_m(T(z_m)) \in C,\: \omega_m = \xi_m(T(z_m)) \in \sn$$
and passing again to a subsequence if necessary, we may assume that $y_m \to z_0 \in C,\: \omega_m \to \omega_0 \in \sn.$ 
Then for the generalized bicharacteristic $\delta_{\mu} (t) = (t, x(t), 1, \xi(t))$ issued from $\mu = (0, z_0, 1, \omega_0)$ we have $T(\delta_{\mu}) = \infty.$ Next, consider the hyperplane $Z_{\omega_0}$ passing through $z$ and orthogonal to $\omega_0$ and the set of points $u \in Z_{\infty}$ such that the generalized bicharacteristics $\gamma_u$ issued from $u \in Z_{\infty}$ with 
direction $\omega_0$ satisfies the condition $T(\gamma_u) = \infty.$ The set $Z_{\infty}\cap Z_{\omega_0}$ is closed in $Z_{\omega_0}$ and $Z_{\omega_0} \setminus Z_{\infty} \neq \emptyset.$ Repeating the above argument, we obtain rays $\gamma_m$ with sojourn times $T_{\gamma_m} \longrightarrow +\infty.$ Using Proposition \ref{3.2}, we may assume that each ray $\gamma_m$ is unbounded
in both directions, i.e. $\gamma_m$ is an $(\omega_0, \theta_m)$-ray for some $\theta_m\in \sn$. Moreover, according to results in  [1] and [16], these rays can be approximated by ordinary reflecting ones, so we may assume that 
each $\gamma_m$ is an ordinary reflecting $(\omega_m,\theta_m)$-ray for some
$\omega_m,\theta_m \in \sn$. This completes the proof.
\end{proof}

\def\ow{{\mathcal O}(W)}

To show that the rays $\gamma_m$ constructed in Theorem \ref{3.1} produce singularities, we need to check the condition $(\ref{eq:2.4})$. In general the ordinary reflecting ray $\gamma_m$ could be degenerate and we have  to replace $\gamma_m$ by another ordinary reflecting non-degenerate $(\theta_m', \: \omega_m')$-ray $\gamma_m'$ with sojourn time $T_m'$ sufficiently close to $T_{\gamma_m}.$ Our argument concerns the rays issued from a small neighborhood $W \subset C \times \sn$ of the point $(z_0, \omega_0) \in C \times \sn$ introduced in the proof of Theorem 3.1.\\

Let ${\mathcal O}(W)$ be  the set of all pairs of directions $\ot \in \ssn$ such that there exists an ordinary reflecting  $\ot$-ray issued 
from $(x, \omega) \in W$ with {\it outgoing} direction $\theta \in \sn.$ To obtain convenient approximations with $\ot$-rays issued from $W$, 
it is desirable to know that $\ow$ has a positive measure in $\ssn$ for all sufficiently small neighborhoods $W \subset C \times \sn$ of $(z_0, \omega_0)$. Roughly speaking this means that the trapping generalized bicharacteristic $\delta_{\mu}(t)$ introduced above is non-degenerate in some sense. More precisely, we introduce the following

\begin{defn}\label{def:1}
The generalized bicharacteristic $\gamma$  issued from  $(y, \eta) \in C \times \sn$ is called weakly non-degenerate if for every neighborhood $W \subset C \times \sn$ of $(y, \eta)$ the set $\ow$ has a positive measure in $\ssn.$
\end{defn}

The above definition generalizes that of a non-degenerate ordinary reflecting ray $\gamma$ given in Section 2. Indeed, let $\gamma$ 
be an ordinary reflecting non-degenerate $(\omega_0, \theta_0)$-ray issued from $(x_0, \omega_0) \in C \times \sn$. Let $Z = Z_{\omega_0}$ and  consider the $C^{\infty}$ map
$$D = X \times \Gamma \ni (x, \omega) \longrightarrow f(x, \omega) \in \sn,$$
where $X \subset Z$ is a small neighborhood of $x_0$, $\Gamma \subset \sn$ is a small neighborhood of $\omega_0$, and $f(x, \omega)$ is the outgoing direction of the ray issued from $x$ in direction $\omega$. We have $\det f'_x(x_0, \omega_0) \neq 0$ and we may assume that  $D$ is chosen small enough so that $\det f'_x(x, \omega) \neq 0$ for $(x, \omega) \in \bar{D}.$ Set 
$$\max_{(x, \omega) \in \bar{D}} \|(f_x'(x, \omega))^{-1}\| = \frac{1}{\alpha}.$$
Then for small $\epsilon > 0$ we have $\|f_x'(x, \omega) - f_x'(x_0, \omega_0) \| \leq \frac{\alpha}{4}$, provided 
$\|x - x_0\| < \epsilon,\: \|\omega - \omega_0\| < \epsilon$ and 
$$X_{\epsilon} = \{ x \in Z:\: \|x - x_0\| < \epsilon \} \subset X,\:\: \Gamma_{\epsilon} = \{\omega \in \sn:\:\|\omega - \omega_0\| < \epsilon\} \subset \Gamma.$$
Next consider the set
$$\Xi_{\epsilon} = \{ \theta \in \sn:\: \|\theta - \theta_0\| < \frac{\epsilon \alpha}{4}\}.$$
Then taking $\epsilon' \in (0, \epsilon)$ so that $\|f(x_0, \omega) - \theta_0\| < \frac{\epsilon \alpha}{4}$ for 
$\omega \in \Gamma_{\epsilon'}$, and applying the inverse mapping theorem (see Section 5 in \cite{PS2}), we conclude that for every fixed $\omega \in \Gamma_{\epsilon'}$ and every fixed $\theta \in \Xi_{\epsilon'}$ we can find $x_{\ot} \in X_{\epsilon}$ with 
$f(x_{\ot}, \omega) = \theta.$ Consequently,  the corresponding set of directions 
$ \Gamma_{\epsilon'} \times \Xi_{\epsilon'} \subset {\mathcal O}(W)$ has positive measure in $\ssn$.  
This argument works for every neighborhood of $(x_0, \omega_0)$, so $\gamma$ is weakly non-degenerate according to Definition \ref{def:1}.\\

\begin{rem} In general a weakly non-degenerate ordinary reflecting ray does not need to be non-degenerate. To see this, first notice that the set of those $(y, \eta) \in C\times \sn$ that generate
weakly non-degenerate bicharacteristics is closed in $C\times \sn$. Now consider the special case when $K$ is convex with vanishing Gauss curvature at some point $x_0 \in \partial K$ and strictly positive Gauss curvature at any other point of $\partial K$. Consider a reflecting ray $\gamma$ in $\R^n$ with a single
reflection point at $x_0$. Then, as is well-known,  $\gamma$ is degenerate, that is the differential cross section vanishes. However, arbitrarily close to $\gamma$ we can choose an ordinary reflecting ray $\delta_m$ with a single reflection point $x_m \neq x_0$. Then $\delta_m$ is non-degenerate and hence it is weakly non-degenerate. Thus, $\gamma$ can be approximated arbitrarily well with weakly non-degenerate rays, and therefore $\gamma$ itself is weakly non-degenerate.
\end{rem}
 
Now we have a stronger version of Theorem \ref{3.1}.

\begin{thm} Let the obstacle $K$  have at least one trapping weakly non-degenerate bicharacteristic $\delta$ issued from $(y, \eta) \in C \times \sn$ and let $K$ satisfy $({\mathcal G}).$ Then there exists a sequence of ordinary reflecting non-degenerate $(\omega_m, \theta_m)$-rays  $\gamma_m$ with sojourn times $T_{\gamma_m} \longrightarrow \infty.$
\end{thm}

\begin{proof} Let $W_m \subset C \times \sn$ be a neighborhood of $(y, \eta)$ such that for every $z \in W_m$ the generalized bicharacteristic $\gamma_z$ issued from $z$ satisfies the condition $T(\gamma_z) > m$. The continuity of the compressed generalized flow guarantees the existence of $W_m$ for all $m \in \N$. Moreover, we have $W_{m+1} \subset W_m.$ 
Consider the {\it open subset} $F_m$ of $C \times \sn \times C \times \sn$ consisting of those
$(x, \omega, z, \theta)$ such that $(x, \omega) \in W_m$ and there exists an ordinary reflecting $\ot$-ray
 issued from $ (x, \omega) \in W_m$ and passing through $z$ with direction $\theta$.

The projection $F_m \ni (x, \omega, z, \theta) \longrightarrow (\omega, \theta)$ is smooth and Sard's theorem implies the existence of a set ${\mathcal D}_m \subset \ssn$ with measure zero so that if $\ot \notin {\mathcal D}_m$ the corresponding $\ot$-ray issued from $(x, \omega) \in W_m$  is non-degenerate. Then the set ${\mathcal O}(W_m) \setminus {\mathcal D}_m$ has a positive measure and taking $(\omega_m, \theta_m) \in  {\mathcal O}(W_m) \setminus {\mathcal D}_m$ we obtain an ordinary reflecting non-degenerate $(\omega_m,\: \theta_m)$-ray $\delta_m$ with sojourn time $T_m$ issued from $z_m \in W_m$. Next we choose  
$$q(m) > \max\{ m+1, T_m\},\: q(m) \in \N$$
 and repeat the same argument for $W_{q(m)}$ and $F_{q(m)}.$ This completes the proof.
\end{proof}

\begin{rem} In general,  a generalized trapping ray $\delta$ can be weakly degenerate if its reflection points lie on flat regions of the boundary. In the case when $K$ is a finite disjoint union of several convex domains sufficient conditions  for a trapping ray to be weakly non-degenerate are given in \cite{PS3}. On the other hand, we expect that the sojourn time $T_{\gamma}$ of an ordinary reflecting ray $\gamma$ may  produce a singularity of the scattering kernel if the condition  $(\ref{eq:2.4})$ is replaced by some weaker one. For this purpose one needs a generalization of Theorem 2.1 based on the asymptotics of oscillatory integrals with degenerate critical points.
\end{rem}
 
Now assume that $\gamma$ is an ordinary reflecting non-degenerate $\ot$-ray with sojourn time $T_{\gamma}$ issued from $(x, \xi) \in C \times \sn.$ For a such ray the condition $(\ref{eq:2.4})$ is not necessarily fulfilled. Since $\gamma$ is non-degenerate, there are no $\ot$-rays $\delta$ with sojourn time $T_{\gamma}$ issued from points in a small neighborhood of $(x, \xi).$ This is not sufficient for $(\ref{eq:2.4})$ and we must take into account all $\ot$-rays. The result in \cite{St1} says that for almost all directions $\ot \in \ssn$ all $\ot$-rays are reflecting ones and the result in \cite{PS2} implies the property $(\ref{eq:2.3})$ for the sojourn times of ordinary reflecting rays $\ot$-ray, provided that $\ot$ is outside some set of measure zero. Thus we can approximate $\ot$ by directions $(\omega', \theta')$ for which the above two properties hold. Next, the fact that $\gamma$ is non-degenerate combined with the inverse mapping theorem make possible to find an ordinary reflecting non-degenerate $(\omega', \theta')$-ray $\gamma'$ with sojourn time $T_{\gamma}'$ sufficiently close to $T_{\gamma}$ so that $(\ref{eq:2.3})$ and $(\ref{eq:2.4})$ hold for $\gamma'$. We refer to Section 5 in \cite{PS2} for details concerning the application of the inverse mapping theorem.
Finally, we obtain the following

\begin{thm}  Under the assumptions of Theorem $3.6$ there exists a sequence $(\omega_m, \theta_m) \in \ssn$ and ordinary reflecting non-degenerate  $(\omega_m, \theta_m)$-rays $\gamma_m$ with sojourn times $T_m \longrightarrow \infty$ so that
\begin{equation} \label{eq:s}
-T_m \in {\rm sing}\: {\rm supp}\: s(t, \omega_m, \theta_m),\:\: \forall m \in \N.
\end{equation}
\end{thm}

The relation $(\ref{eq:s})$ was called property $(S)$ in \cite{PS2} and it was conjectured that every trapping obstacle has the property $(S)$. The above result says that this is true if the generalized Hamiltonian flow is continuous and if there is at least one weakly non-degenerate trapping ray $\delta$. The assumption that $\delta$ is weakly non-degenerate has been omitted in Theorem 8 in \cite{PS4}. 

\section{Trapping rays and estimates of the scattering amplitude}

The scattering resonances are related to the behavior of the modified resolvent of the Laplacian.  For $\Im \lambda <  0$ consider the outgoing resolvent $R(\lambda) = (-\Delta - \lambda^2)^{-1}$ 
of the Laplacian in $\Omega$ with Dirichlet boundary conditions on $\partial K.$  The outgoing condition means that for $f \in C_0^{\infty}(\Omega)$ there exists $g(x) \in C_0^{\infty}(\R^n)$ so that we have 
$$R(\lambda) f(x) = R_0(\lambda) g(x),\:\: |x| \to \infty,$$
where 
$$R_0(\lambda) = (-\Delta - \lambda^2)^{-1}: L^2_{\rm comp} (\R^n) \longrightarrow H^2_{\rm loc}(\R^n)$$
 is the outgoing resolvent of the free Laplacian in $\R^n$. The operator
$$R(\lambda) : L_{\rm comp}^2(\Omega) \ni f \longrightarrow R(\lambda) f \in  H^2_{{\rm loc}}(\Omega)$$
has a meromorphic continuation in $\C$ with poles $\lambda_j, \: \Im \lambda_j > 0$, called {\it resonances} (\cite{LP}). Let $\chi \in C_0^{\infty}(\R^n)$ be a cut-off function such that $\chi(x) = 1$ on a neighborhood of $K$. It is easy to see that the {\it modified resolvent}
$$R_{\chi}(\lambda) = \chi R(\lambda)\chi$$ 
has a meromorphic continuation in $\C$ and the poles of $R_{\chi}(\lambda)$ are independent of the choice of $\chi$. These poles coincide with their multiplicities with those of the resonances. On the other hand, the scattering amplitude $a(\lambda, \theta, \omega)$ also admits a meromorphic continuation in $\C$ and the poles of this continuation and their multiplicities are the same as those of the resonances (see \cite{LP}). From the general results on propagation of singularities given in \cite{MS}, it follows that if $K$ is non-trapping, there exist $\epsilon > 0$ and $d > 0$ so that $R_{\chi}(\lambda)$ has no poles in the domain
$$U_{\epsilon, d} = \{ \lambda \in \C:\:0 \leq \Im \lambda \leq  \epsilon \log (1 + |\lambda|) -d\}.$$
Moreover, for non-trapping obstacles we have the estimate (see \cite{Va})
$$\|R_{\chi}(\lambda)\|_{L^2(\Omega) \longrightarrow L^2(\Omega)} \leq \frac{C}{| \lambda|} e^{C |\Im \lambda|},\: \forall \lambda \in U_{\epsilon, d}.$$

We conjecture that the existence of singularities $t_m \longrightarrow -\infty$ of the scattering kernel $s(t, \theta_m, \omega_m)$ implies that for every $\epsilon > 0$ and $d > 0$ we have resonances in  $U_{\epsilon, d}$. 

Here we prove a weaker result assuming an estimate of the scattering amplitude.

\begin{thm} Suppose that there exist $m \in \N, \: \alpha \geq 0, \: \epsilon > 0,\: d > 0$ and $C > 0$ so that $a(\lambda, \theta, \omega)$ is analytic in $U_{\epsilon, d}$ and 
\begin{equation} \label{eq:4.1}
|a(\lambda, \theta, \omega)| \leq C (1 + |\lambda|)^m e^{\alpha |\Im \lambda|},\:\:\forall \ot \in \ssn,\:\: \forall \lambda \in U_{\epsilon, d}.
\end{equation}
Then if $K$ satisfies $({\mathcal G})$, there are no trapping weakly non-degenerate rays in $\Omega.$
\end{thm}

The proof of this result follows directly from the statement in Theorem 2.3 in \cite{PS2}. In fact, if 
there exists a weakly non-degenerate trapping ray, we can apply Theorem 3.7, and for the sequence of sojourn times 
$\{-T_m\},\: T_m \to \infty,$ related to a weakly non-degenerate ray $\delta$, an application of Theorem 2.1 yields a sequence of delta type isolated singularities of the scattering kernel. The existence of these singularities combined with the estimate (\ref{eq:4.1}) leads to a contradiction since we may apply the following

\begin{lem} $($\cite{PS2}$)$ Let $u \in {\mathcal S}'(\R)$ be a distribution. Assume that the Fourier transform $\hat{u}(\xi), \xi \in \R,$ admits an analytic continuation in 
$$W_{\epsilon, d} = \{\xi \in \C:\:d - \epsilon \log (1 + |\xi|) \leq \Im \xi \leq 0\},\: \epsilon > 0, \: d > 0$$
such that for all $\xi \in W_{\epsilon, d}$ we have
$$|\hat{u}(\xi)| \leq C(1 + |\xi|)^N e^ {\gamma |\Im \xi|},\: \gamma \geq 0.$$
Then for each $q \in \N$ there exists $t_q < \tau$ and $v_q \in C^q(\R)$ such that $u = v_q$ for $t \leq t_q.$
\end{lem}
Here the Fourier transform $\hat{u}(\xi) = \int e^{-\ii t \xi} u(t)
dt$ for $u \in C_0^{\infty}(\R)$ and for $\lambda \in \R$ we have

$$\hat{s}(\lambda, \theta, \omega) = \Bigl(\frac{\ii \lambda}{2 \pi}\Bigr)^{(n-1)/2}\overline{a(\lambda, \theta,
  \omega)} = \Bigl(\frac{\ii \lambda}{2 \pi}\Bigr)^{(n-1)/2}a(-\lambda, \theta, \omega).$$
 Thus $\hat{s}(\lambda,
\theta, \omega)$ admits an analytic continuation in $W_{\epsilon, d}$ and the estimate (\ref{eq:4.1}) implies an estimate for $\hat{s}(\lambda, \theta, \omega)$ in $W_{\epsilon, d}.$\\

 It is easy to see that the analyticity of $R_{\chi}(\lambda)$ in $U_{\epsilon, d}$ and the estimate
\begin{equation} \label{eq:4.2}
\|R_{\chi}(\lambda)\|_{L^2(\Omega) \longrightarrow L^2(\Omega)} \leq  C' (1 + |\lambda|)^{m'} e^{\alpha' |\Im \lambda|},\:\ \forall \lambda \in U_{\epsilon, d}
\end{equation}
with  $m' \in \N, \: \alpha' \geq 0$, imply $(\ref{eq:4.1})$ with suitable $m$ and $\alpha.$ This follows from the representation of the scattering amplitude involving the cut-off resolvent $R_{\psi}(\lambda)$ (see \cite{PS2}, \cite{PZ}) with $\psi \in C_0^{\infty} (\R^n)$ having support in $\{x \in \R^n:\:0 < a' \leq |x| \leq b'\}.$ Moreover, we can take $a' < b'$ arbitrary large. More precisely, let $\varphi_a \in C_0^{\infty}(\R^n)$ be a cut-off function such that
$\varphi_a(x) = 1$ for $|x| \leq \rho.$ Set 
$$F_a(\lambda, \omega) = [\Delta \varphi_a + 2 \ii \lambda \langle \nabla \varphi_a, \omega \rangle ] e^{\ii \lambda \langle x, \omega \rangle}.$$ 
Let $\varphi_b(x) \in C_0^{\infty}(\R^n)$ be such that $\varphi_b(x) = 1$ on a neighborhood of $K$ and $\varphi_a(x) = 1$ on supp $\varphi_b.$ The scattering amplitude $a(\lambda, \theta, \omega)$ has the representation
$$ a(\lambda, \theta, \omega) = c_n \lambda^{(n-3)/2} \int_{\Omega} e^{-\ii \lambda \la x, \theta \ra} \Bigl[ (\Delta \varphi_b)R(\lambda) F_a(\lambda, \omega) $$
$$+ 2 \la \nabla_x \varphi_b, \nabla_x (R(\lambda) F_a(\lambda, \omega))\ra \Bigr] dx$$
with a constant $c_n$ depending on $n$ and this representation is independent of the choice of $\varphi_a$ and
$\varphi_b.$ In particular, if the estimate $(\ref{eq:4.2})$ holds, then the obstacle $K$ has no trapping weakly non-degenerate rays.\\

Consider the cut-off resolvent $R_{\psi}(\lambda)$ with supp $\psi
\subset \{x \in \R^n: \:0 < a' < |x| < b'\}.$ For $\lambda \in \R$ and sufficiently large $a'$ and $b'$ N. Burq \cite{B} (see also \cite{CV}) established the estimate 
\begin{equation} \label{eq:4.3}
\|R_{\psi}(\lambda)\|_{L^2(\Omega) \longrightarrow L^2(\Omega)} \leq \frac{C_2}{1 + |\lambda|},\: \lambda \in \R
\end{equation}
without any geometrical restriction of $K$. On the other hand, if we have resonances converging sufficiently fast to the real axis, the norm $\|R_{\chi}(\lambda)\|_{L^2(\Omega) \to L^2(\Omega)}$ with $\chi = 1$ on $K$ increases like ${\mathcal O}(e^{\C|\lambda|})$ for $\lambda \in \R,\: |\lambda| \to \infty.$ Thus the existence of trapping rays influences the estimates of $R_{\chi}(\lambda)$ with $\chi(x)$ equal to 1 on a neighborhood of the obstacle and the behaviors of the scattering amplitude $a(\lambda, \theta, \omega)$ and the cut-off resolvent $R_{\chi}(\lambda)$ for $\lambda \in \R$ are rather different if we have trapping rays.\\

 It is interesting to examine the link between the estimates for $a(\lambda, \theta, \omega)$ and the cut-off resolvent $R_{\chi}(\lambda)$ for $\lambda \in U_{\epsilon, d}.$ In this direction we have the following
\begin{thm} Under the assumptions of Theorem $4.1$ for $a(\lambda, \theta, \omega)$ the cut-off resolvent $R_{\chi}(\lambda)$ with arbitrary $\chi \in C_0^{\infty}(\R^n)$ satisfies the estimate $(\ref{eq:4.2})$ in $U_{\epsilon, d}$ with suitable $C' > 0,\:m' \in \N$ and $\alpha' \geq 0.$
\end{thm}

\begin{proof} The poles of $a(\lambda, \theta, \omega)$ in $\{z \in \C: \: \Im \lambda > 0\}$  coincide with the poles of the scattering operator 
$$S(\lambda) = I + K(\lambda): L^2(\Ss^{n-1}) \longrightarrow L^2(\Ss^{n-1}),$$
where $K(\lambda)$ has kernel $a(\lambda, \theta, \omega).$
Thus the estimate (\ref{eq:4.1}) of $a(\lambda, \theta, \omega)$ leads to  an estimate of the same type for the norm of the scattering operator $S(\lambda)$ for $\lambda \in U_{\epsilon, d}.$
 Notice that $S^{-1}(\lambda) = S^*(\bar{\lambda})$ for every $\lambda \in \C$ for which the operator $S(\lambda)$ is invertible. Moreover, the resonances $\lambda_j$ are symmetric with respect to the imaginary axe $\ii \R$.\\

 Consider the energy space  $H =H_D(\Omega) \oplus L^2(\Omega)$, the unitary group $U(t) = e^{\ii t G}$ in $H$ related to the Dirichlet problem for the wave equation in $\Omega$ and the semigroup $Z^b(t) = P_{+}^b U(t) P_{-}^b,\: t \geq 0,$ introduced by Lax and Phillips  (\cite{LP}). Here $P_{\pm}^b$ are the orthogonal projections on the orthogonal complements of the Lax-Phillips spaces $D_{\pm}^b,\: b > \rho$ (see \cite{LP} for the notation). Let $B^b$ be the generator of $Z^b(t)$.
The eigenvalues $z_j$ of $B^b$ are independent of $b$, the poles of the scattering operator $S(\lambda)$ are $\{-\ii z_j  \in \C,\: z_j \in \:{\rm spec}\: B^b\}$ and the multiplicities of $z_j$ and $-\ii z_j$ coincide. Given a fixed function $\chi \in C_0^{\infty}(\R^n)$, equal to 1 on $K$, we can choose $b > 0$ so that $P_{\pm}^b\chi = \chi P_{\pm}^b = \chi.$ We fix $b > 0$ with this property and will write below $B, \: P_{\pm}$ instead of $B^b, \: P_{\pm}^b$. Changing the outgoing representation of $H$, we may introduce another scattering operator
$S_1(\lambda)$ (see Chapter III in \cite{LP}) which is an operator-valued inner function in $\{\lambda \in \C: \: \Im \lambda \leq 0\}$ and
\begin{equation} \label{eq:4.4}
\|S_1(\lambda)\|_{L^2(\Ss^{n-1}) \to L^2(\Ss^{n-1})} \leq 1,\: \Im \lambda \leq 0.
\end{equation}
The estimates (\ref{eq:4.4}) is not true for the scattering operator $S(\lambda) = I + K(\lambda)$ related to the scattering amplitude. On the other hand, the link between the outgoing representations of $H$ introduced in Chapters III and V in \cite{LP} implies the equality
\begin{equation}  \label{eq:4.5}
S_1(\lambda) = e^{-\ii \beta \lambda}S(\lambda), \: \beta > 0.
\end{equation}
The following estimate established in Theorem 3.2 in \cite{LP} plays a crucial role
$$\|(\ii \lambda - B)^{-1}\|_{H \to H} \leq \frac{3}{2|\Im \lambda|}\|S_1^{-1}(\bar{\lambda})\|_{L^2(\sn) \to L^2(\sn)}, \: \forall \lambda \in U_{\epsilon, d} \setminus \R.$$

Since $S^{-1}(\lambda) = S^*(\bar{\lambda})$ for all $\lambda \in \C$
for which $S(\lambda)$ is invertible, the estimates (\ref{eq:4.1}) and (\ref{eq:4.5}) imply
$$\|(\ii \lambda - B)^{-1}\|_{H \to H} \leq \frac{3 e^{\beta|\Im \lambda|}}{2|\Im \lambda|}\|S(\lambda)\|_{L^2(\sn) \to L^2(\sn)}$$
$$\leq C_1 (1 + |\lambda|)^{m'} \frac{e^{\alpha' |\Im \lambda|}}{|\Im \lambda|},\: \forall \lambda \in U_{\epsilon, d} \setminus \R.$$
For $\Re \lambda > 0$ we have
$$\chi(\lambda - B)^{-1}\chi = \int_0^{\infty} e^{-\lambda t} \chi P_{+}U(t)P_{-}\chi dt = -\ii \chi (-\ii \lambda -  G)^{-1}\chi$$
and by an analytic continuation we obtain this equality for $ \lambda \in \ii U_{\epsilon, d}.$
By using the relation between $R_{\chi}(\lambda)$ and $\chi(\lambda - G)^{-1}\chi$, we deduce the estimate 
$$\|R_{\chi}(\lambda)\|_{L^2(\Omega) \to L^2(\Omega)} \leq C_2 (1 + |\lambda|)^{m'} \frac{e^{\alpha' |\Im \lambda|}}{|\Im \lambda|}$$
for $\Im \lambda = \epsilon \log(1 + |\lambda|) - d,\: |\Re \lambda| \geq c_0.$ On the other hand, $\|R_{\chi}(\lambda)\|_{L^2(\Omega) \to L^2(\Omega)}$ is bounded for $\Im \lambda = -c_1 < 0$ and have the estimate (see for example \cite{TZ})
$$\|R_{\chi}(\lambda)\|_{L^2(\Omega) \to L^2(\Omega) } \leq
Ce^{C|\lambda|^n}, \: \Im \lambda \leq \epsilon \log(1 + |\lambda) - d.$$
Then an application of the Pragmen-Lindel\"of theorem yields the result.

\end{proof}
It is an interesting open problem to show that the analyticity of  $a(\lambda, \theta, \omega)$ in $U_{\epsilon, d}$ implies the estimate $(\ref{eq:4.1})$ with suitable $m,\: \alpha$ and $C$ without any information for the {\it geometry} of the obstacle. The same problem arises for the strip $V_{\delta} =\{\lambda \in \C: \: 0 \leq \Im \lambda \leq \delta\}$
and we have the following\\

{\bf Conjecture.} {\it Assume that the scattering amplitude
  $a(\lambda, \theta, \omega)$ is analytic in $V_{\delta}.$ Then there
  exists a constants $C_1 > 0, \: C \geq 0$ such that}
$$|a(\lambda, \theta, \omega)| \leq C_1 e^{C|\lambda|^2},\: \forall \ot \in \ssn,\:\forall \lambda \in V_{\delta}.$$
For $n = 3$ this conjecture is true since we may obtain an exponential
estimate ${\mathcal O}(e^{C|\lambda|^2})$ for the cut-off resolvent
$R_{\chi}(\lambda),\: \lambda \in V_{\delta}$ (see for more details \cite{BP}).

\end{document}